\documentclass[12pt,a4paper]{article}

\usepackage[latin1]{inputenc}
\usepackage{amsmath,amssymb,amsthm,srcltx}
\usepackage{mathtools}
\usepackage[dvips]{color}
\usepackage[margin]{fixme}
\usepackage{enumerate}
\usepackage{stmaryrd}
\usepackage[colorinlistoftodos]{todonotes}
\usepackage{bbm}
\usepackage{dsfont}
\numberwithin{equation}{section}
\usepackage{fullpage}

\newcommand\R{\mathbb{R}}
\newcommand\Q{\mathbb{Q}}
\DeclareMathOperator\p{\mathbb{P}}
\DeclareMathOperator\sign{\mathrm{sign}}
\newcommand\N{\mathbb{N}}

\newcommand\F{\mathcal{F}}

\newcommand\E{\mathbb{E}}

\renewcommand\i{\infty}

\newcommand\ve{\varepsilon}
\newcommand{\vp}{\varphi}
\newcommand{\cA}{\mathcal{A}}
\newcommand{\cB}{\mathcal{B}}
\newcommand{\cZ}{\mathcal{Z}}
\newcommand{\cC}{\mathcal{C}}
\newcommand{\cD}{\mathcal{D}}

\newcommand{\hvp}{\widehat\varphi}
\newcommand{\hY}{\widehat Y}
\newcommand{\hS}{\widehat S}

\renewcommand\div{\mathbin/}

\renewcommand\[{\begin{equation}}
\renewcommand\]{\end{equation}}
\newcommand\tsfrac[2]{\bgroup\textstyle\frac#1#2\egroup}

\newcommand{\sint}{\stackrel{\mbox{\tiny$\bullet$}}{}}

\newcommand{\Var}{\operatorname{Var}} 

\DeclarePairedDelimiter\abs{\lvert}{\rvert}
\let\brack\undefined
\DeclarePairedDelimiter\brack\lbrack\rbrack
\DeclarePairedDelimiter\paren{(}{)}

\DeclareMathOperator\Cov{Cov}

\newtheorem{theorem}{Theorem}[section]

\newtheorem{corollary}[theorem]{Corollary}

\newtheorem{definition}[theorem]{Definition}
\newtheorem{lemma}[theorem]{Lemma}
\newtheorem{proposition}[theorem]{Proposition}

\theoremstyle{definition}

\newtheorem{remark}[theorem]{Remark}

\newcommand{\Om}{\Omega}
\newcommand{\om}{\omega}
\newcommand{\bt}{\begin{thm}}
\newcommand{\et}{\end{thm}}
\newcommand{\br}{\begin{remark}}
\newcommand{\er}{\end{remark}}
\newcommand{\bl}{\begin{lemma}}
\newcommand{\el}{\end{lemma}}
\newcommand{\bp}{\begin{proof}}
\newcommand{\ep}{\end{proof}}
\newcommand{\bal}{\begin{align*}}
\newcommand{\eal}{\end{align*}}
\newcommand{\bi}{\begin{itemize}}
\newcommand{\be}{\begin{equation}}
\newcommand{\ee}{\end{equation}}
\newcommand{\bea}{\begin{eqnarray}}
\newcommand{\eea}{\end{eqnarray}}
\newcommand{\ba}{\begin{align*}}
\newcommand{\ea}{\end{align*}}
\newcommand{\ei}{\end{itemize}}
\newcommand{\vt}{\vartheta}
\newcommand{\hvt}{\widehat\vartheta}
\newcommand{\hg}{\widehat g}

\newcommand{\Lim}{\lim\limits}
\newcommand{\Inf}{\inf\limits}

\begin{document}

\title{Shadow prices, fractional Brownian motion, and portfolio optimisation under transaction costs}
\author{Christoph Czichowsky\footnote{Department of Mathematics, London School of Economics and Political Science, Columbia House, Houghton Street, London WC2A 2AE, UK, {\tt c.czichowsky@lse.ac.uk}.}
        \and R\'emi Peyre\footnote{CNRS \& Institut \'Elie Cartan de Lorraine, Campus Aiguillettes, 54506 Vand\oe uvre-l\`es-Nancy \textsc{cedex}, France, \texttt{remi.peyre@univ-lorraine.fr}.}
 \and Walter Schachermayer\footnote{Fakult\"at f\"ur Mathematik, Universit\"at Wien, Oskar-Morgenstern-Platz 1, A-1090 Wien, Austria, {\tt walter.schachermayer@univie.ac.at} and the Institute for Theoretical Studies, ETH Zurich. Partially supported by the Austrian Science Fund (FWF) under grant P25815, the Vienna Science and Technology Fund (WWTF) under grant MA09-003 and Dr. Max R\"ossler, the Walter Haefner Foundation and the ETH Zurich Foundation.}
\and Junjian Yang\footnote{Fakult\"at f\"ur Mathematik, Universit\"at Wien, Oskar-Morgenstern-Platz 1, A-1090 Wien, Austria, {\tt junjian.yang@univie.ac.at}. Financial support by the Austrian Science Fund (FWF) under grant P25815 is gratefully acknowledged.}}

\date{\today}
\maketitle

\begin{abstract}\noindent 
We continue the analysis of our previous paper \cite{CSY15} pertaining to the existence of a shadow price process for portfolio optimisation under proportional transaction costs. There, we established a positive answer for a continuous price process $S=(S_t)_{0\leq t\leq T}$ satisfying the condition $(NUPBR)$ of ``no unbounded profit with bounded risk''. This condition requires that $S$ is a semimartingale and therefore is too restrictive for applications to models driven by fractional Brownian motion. In the present paper, we derive the same conclusion under the weaker condition $(TWC)$ of ``two way crossing'', which does not require $S$ to be a semimartingale. 
Using a recent result of R.~Peyre, this allows us to show the existence of a shadow price for exponential fractional Brownian motion and \emph{all} utility functions defined  on the positive half-line having reasonable asymptotic elasticity. Prime examples of such utilities are logarithmic or power utility. 
\end{abstract}

\vspace{-0.2cm}\newline
\vspace{-0.2cm}\newline

\section{Introduction}
In mathematical finance, one classically works with so-called \emph{frictionless financial markets}, where at each time $t$ arbitrary amounts of stock can be bought and sold at the same price $S_t$. Here, the mathematical structure of utility maximisation essentially implies that an optimal trading strategy only exists, if the discounted price processes $S=(S_t)_{0\leq t\leq T}$ of the underlying financial instruments are so-called \emph{semimartingales}, that is, stochastic processes which are ``good integrators'' (see \cite{AI05,LZ08,KP08}). The latter is also related to absence of arbitrage opportunities in frictionless markets either in the form of ``no free lunch with vanishing risk'' $(NFLVR)$ (see Theorem 7.2 of \cite{DS94}) or its local version of ``no unbounded profit with bounded risk'' $(NUPBR)$ (see Theorem 2.3 of \cite{KP08}) and explains why most of the literature assumes that discounted prices are semimartingales.

While the semimartingale property allows to employ the powerful tools of It\^o calculus to obtain optimal trading strategies in frictionless financial markets, it rules out non-semimartingale models based on fractional Brownian motion. These models have been proposed by Mandelbrot~\cite{M63} about fifty years ago. Their fractional scaling and related statistical properties distinguish them as a natural class of price processes beyond the traditional semimartingale setup.

For frictionless trading, Rogers~\cite{R97} and Cheridito~\cite{C03} show how to exploit the non-semimartingality of fractional models like the \emph{fractional Black-Scholes model} 
  $$S_t=\exp(\mu t + \sigma B^H_t),$$ 
where $\mu\in\R$, $\sigma>0$ and $B^H=(B^H_t)$ is a fractional Brownian motion with Hurst parameter $H\in(0,1)\setminus\frac{1}{2}$, to explicitly construct ``arbitrage opportunities''. In general, this assertion follows from the fact that for locally bounded, c\`adl\`ag, adapted processes ``no free lunch with vanishing risk from simple trading strategies'' implies the semimartingale property (see Theorem 7.2 of \cite{DS94}). While the fractional models provide arbitrage opportunities for frictionless trading, Guasoni \cite{G06} proves that they are arbitrage-free as soon as \emph{proportional transaction costs} are taken into account. Conceptually, this allows to use these models as price processes for portfolio optimisation under transaction costs, as illustrated by Guasoni \cite{G02}. He shows that optimal trading strategies exist for non-semimartingale models under transaction costs, if they are arbitrage free and the indirect utility is finite. 

In this paper, we give a quite satisfactory answer to the existence of a so-called \emph{shadow price} for portfolio optimisation under transaction costs in the fractional Black-Scholes model. This is a semimartingale price process $\hS=(\hS_t)_{0\leq t\leq T}$ taking values in the bid-ask spread such that frictionless trading for that price process leads to the same optimal trading strategy and utility as in the original problem under transaction costs. We show that a shadow price exists for the fractional Black-Scholes model for arbitrary utility functions $U:(0,\infty)\to\R$  on the positive half-line satisfying the condition of reasonable asymptotic elasticity.

For utility functions $U:(0,\infty)\to\R$, we established in \cite{CSY15} the existence of a shadow price under the assumption that $S$ is continuous and satisfies the condition $(NUBPR)$ without transaction costs. The assumption of $(NUBPR)$ requires that the price process $S$ has to be a semimartingale. It therefore rules out applications to models driven by fractional Brownian motion. In addition, in Proposition 4.1 of \cite{CSY15} we constructed an example of a non-decreasing, continuous, sticky price process such that the optimal trading strategy for the problem of maximising logarithmic utility under transaction costs exists, but there is no shadow price. While the stickiness condition is sufficient to guarantee the existence of a shadow price for continuous price processes and utility functions $U:\R\to\R$ on the whole real line that are bounded from above such as exponential utility (see \cite{CS15}), this assumption is not sufficient for utility functions $U:(0,\infty)\to\R$ on the positive half-line. A closer look at the example reveals that the reason for the non-existence of a shadow price is that the optimal trading strategy holds the maximal admissible leverage. This behaviour can only be optimal because the continuous price process can cross any level only in an upwards direction. 

To ensure the existence of a shadow price, it is sufficient to exclude that the optimal trading strategy takes the maximal leverage. This is done in our first main result (Theorem \ref{t1}) by imposing the condition of ``two way crossing'' $(TWC)$ (Definition \ref{def:twc}). Loosely speaking, this condition requires that, whenever the price process crosses a given level in an upwards direction, it also immediately crosses it in a downward direction and vice versa. The condition $(TWC)$ is, in particular, satisfied by continuous martingales. It has been introduced by Bender \cite{Ben12} in the analysis of ``no simple arbitrage'' (without transaction costs), that is, absence of arbitrage with linear combinations of buy-and-hold strategies. The significance of the condition $(TWC)$ in Theorem \ref{t1} is that it holds in the fractional Black-Scholes model because of the fact that fractional Brownian motion satisfies a law of iterated logarithm at stopping times by a recent result of Peyre \cite{Pey15}. This gives the existence of shadow prices for the fractional Black-Scholes model and utility functions that are bounded from above. To extend this to utility functions that are unbounded from above like logarithmic utility $U(x)=\log(x)$, we need to ensure that the problem is well posed and therefore we have to establish that the indirect utility is finite. Since fractional Brownian motion is neither a Markov process nor a semimartingale, we need different tools than in the frictionless setting in order to achieve this result. Here, we use that in the presence of transaction costs any trading can only be profitable, if there is a sufficient price movement. Exploiting estimates on Gaussian processes, we can bound the expected gains from trading by establishing exponential and Gaussian moments of the fluctuations of fractional Brownian motion of size $\delta>0$ and therefore obtain the finiteness of indirect utility for \emph{any} utility function $U:(0,\infty)\to\R$. This allows us to deduce the existence of a shadow price for the fractional Black-Scholes model and arbitrary utility function $U:(0,\infty)\to\R$ satisfying the condition of reasonable asymptotic elasticity, which is our second main contribution (Theorem \ref{t2}) and a fairly complete answer to this question.

Because of the connection to frictionless financial markets, we can exploit tools from It\^o calculus and known results from portfolio optimisation in frictionless markets under transaction costs by applying them to the shadow price $\hS=(\hS_t)_{0\leq t\leq T}$. From this, we obtain for the fractional Black-Scholes model that the shadow price $\hS$ is given by an It\^o process
\begin{align}\label{int:ito}
d \widehat{S}_t = \widehat{S}_t (\widehat{\mu}_t dt + \widehat{\sigma}_t dW_t), \quad 0 \leq t \leq T,
\end{align}
where $\widehat{\mu}=(\widehat{\mu}_t)_{0 \leq t \leq T}$ and $\widehat{\sigma}=(\widehat{\sigma}_t)_{0 \leq t \leq T}$ are predictable processes such that the solution to \eqref{int:ito} is well-defined in the sense of It\^o integration.

The importance of the existence of a shadow price for logarithmic utility is that the optimal trading strategy to the frictionless problem is \emph{myopic}. That is, it consists of holding a fraction of wealth in the risky asset that is given by the local mean-variance tradeoff of the returns
$$\widehat{\pi}_t=\frac{\widehat{\mu}_t}{\widehat{\sigma}^2_t}, \quad 0\leq t\leq T.$$
By definition of the shadow price, the optimal trading strategy to the frictionless problem for $\hS$ coincides with that to the problem for the original price process $S$ under transaction costs. This implies that 
$$\widehat{\pi}_t=\frac{\widehat{\mu}_t}{\widehat{\sigma}^2_t}=\frac{\hvp^1_{t-}\hS_t}{\hvp^0_{t-}+\hvp^1_{t-}\hS_t},\qquad 0\leq t\leq T,$$
for the optimal trading strategy $\hvp=(\hvp^0_t,\hvp^1_t)_{0_-\leq t\leq T}$ under transaction costs (the notation is taken from \cite{CSY15} and will be recalled later). Therefore, the optimal trading strategy under transaction costs is directly linked to the coefficients $\widehat{\mu}=(\widehat{\mu}_t)_{0 \leq t \leq T}$ and $\widehat{\sigma} = (\widehat{\sigma}_t)_{0 \leq t \leq T}$ of the shadow price process \eqref{int:ito}. We expect that analysing the coefficients $\widehat{\mu}=(\widehat{\mu}_t)_{0 \leq t \leq T}$ and $\widehat{\sigma} = (\widehat{\sigma}_t)_{0 \leq t \leq T}$ should also allow to characterise the optimal trading behaviour under transaction costs in the fractional Black-Scholes model more explicitly similarly as in \cite{GMKS13} for the classical Black-Scholes model. A thorough investigation of this is left to future research.

It is well known that the existence of a shadow price is related to the solution of a suitable dual problem; see \cite{KMK11,CMKS14,CS14,CSY15,CS15}. Under transaction costs, this duality goes back to the seminal work \cite{CK96} of Cvitani\'c and Karatzas and has been subsequently extended to dynamic duality results \cite{CK96,CW01,CS14,CSY15,CS15} for utility functions on the positive half-line as well as static duality results \cite{DPT01,B02,BM03,CO11,BY13} for (possibly) multi-variate utility functions.

To apply this duality in our setup, we need to ensure the existence of so-called \emph{$\lambda$-consistent local martingale deflators}. These processes are used as dual variables similarly as equivalent martingale measures \cite{KLS87,HP91,KLSX91,KS99} and local martingale deflators in the frictionless theory \cite{KK07}. For this, we provide a local version of the fundamental theorem of asset pricing for continuous processes under small transaction costs of \cite{GRS10}. It establishes that, for a continuous price process, the existence of a $\lambda$-consistent local martingale deflator for any size of transaction costs $\lambda\in(0,1)$ is equivalent to having the condition (NOIA) of ``no obvious immediate arbitrage'' (see Definition \ref{def:NOIA}) locally. 

The remainder of the article is organised as follows. We formulate the problem and state our main results in Section \ref{sec:mr}. Their proofs are given in Section \ref{sec:pr}. Section \ref{sec:dt} recalls duality results and provides the local version of the fundamental theorem of asset pricing. In Section \ref{sec:fl}, we establish the exponential and Gaussian moments of the fluctuations of fractional Brownian motion.

\section{Main results}\label{sec:mr}
We consider a financial market consisting of one riskless bond and one risky stock. The riskless asset is assumed to be normalised to one. Trading the risky asset incurs proportional transaction costs $\lambda \in (0,1)$. This means that one has to pay a (higher) ask price $S_t$ when buying risky shares but only receives a lower bid price $(1-\lambda)S_t$ when selling them. Here, \mbox{$S=(S_t)_{0\leq t\leq T}$} denotes a strictly positive, adapted, \emph{continuous} stochastic process defined on some underlying filtered probability space $\big(\Om,\F,(\F_t)_{0\leq t\leq T},\p\big)$ with fixed finite time horizon $T\in(0,\infty)$ satisfying the usual assumptions of right-continuity and completeness. As usual equalities and inequalities between random variables hold up to $\p$-nullsets and between stochastic processes up to $\p$-evanescent sets.

\emph{Trading strategies} are modelled by $\R^2$-valued, c\`adl\`ag and adapted processes $\vp=(\vp^0_t,\vp^1_t)_{0_-\leq t\leq T}$ of finite variation indexed by $[0_-,T]:=\{0_-\}\cup[0,T]$, where $\vp^0_{t}$ and $\vp^1_{t}$ describe the holdings in the riskless and the risky asset, respectively, after rebalancing the portfolio at time $t$. As explained in \cite{CSY15} in more detail, using $[0_-,T]$ instead of $[0,T]$ as index set allows us to use c\`adl\`ag trading strategies. For any process $\psi=(\psi_t)_{0_-\leq t\leq T}$ of finite variation, we denote by $\psi=\psi_{0_-}+\psi^{\uparrow}-\psi^{\downarrow}$ its Jordan-Hahn decomposition into two non-decreasing processes $\psi^{\uparrow}$ and $\psi^{\downarrow}$ starting at zero. 

A trading strategy $\vp=(\vp^0_t,\vp^1_t)_{0_-\leq t\leq T}$ is called \emph{self-financing}, if 
\begin{equation}\label{eq:sf}
\int_s^td\varphi^0_u\leq-\int_s^tS_ud\varphi^{1,\uparrow}_u +\int_s^t (1-\lambda)S_ud\varphi^{1,\downarrow}_u, \quad 0_- \leq s\leq  t \leq T,
\end{equation}
where the integrals can be defined pathwise as a Riemann-Stieltjes integrals.

A self-financing strategy $\vp=(\vp^0,\vp^1)$ is called \emph{admissible}, if its \emph{liquidation value} satisfies
\be
V^{liq}_t(\vp):=\vp^0_t+(\vp^1_t)^+(1-\lambda)S_t-(\vp^1_t)^-S_t\geq 0, \quad 0 \leq t \leq T.\label{liq}
\ee

For $x>0$, we denote by $\cA(x)$ the set of all self-financing and admissible trading strategies under transaction costs $\lambda$ starting from initial endowment $(\vp^0_{0_-},\vp^1_{0_-})=(x,0)$ and
$$\cC(x):=\big\{V^{liq}_T(\vp)~\big|~\vp=(\vp^0,\vp^1)\in\cA(x)\big\}.$$

We consider an economic agent whose goal is to maximise her expected utility from terminal wealth
\begin{equation} \label{P1}
\E[U(g)] \to \max!, \quad g \in \cC(x).
\end{equation}
Here, $U:(0,\infty)\to \R$ denotes an increasing, strictly concave, continuously differentiable utility function satisfying the Inada conditions
\be
\text{$U'(0):=\lim_{x\searrow 0}U'(x) =\infty$ \quad and \quad $U'(\infty):=\lim_{x\nearrow \infty}U'(x)=0$.}\label{Inada}
\ee

In this paper, we continue our analysis of problem \eqref{P1} by using the concept of a shadow price.\begin{definition}\label{def:sp}
A semimartingale price process $\widehat{S}=(\widehat{S}_t)_{0 \leq t \leq T}$ is called a \emph{shadow price process}, if 
\bi
\item[\bf{1)}] $\widehat{S}$ is valued in the bid-ask spread $[(1-\lambda) S, S]$ 
\item[\bf{2)}] The solution $\hvt=(\hvt_t)_{0 \leq t \leq T}$ to the frictionless utility maximisation problem
\begin{equation}\label{P3}
\E\big[U\big(x + \vt \sint \widehat{S}_T\big)\big]\to \max!, \quad \vt \in \mathcal{A}(x;\hS),
\end{equation}
exists  (in the sense of \cite{KS99}), where $\mathcal{A}(x;\hS)$ denotes the set of all self-financing and admissible trading strategies $\vt=(\vt_t)_{0 \leq t \leq T}$ for $\widehat{S}$ without transaction costs. That is, $\widehat{S}$-integrable (in the sense of Itô), predictable processes $\vt=(\vt_t)_{0 \leq t \leq T}$ such that $X_t=x+\vt\sint \hS_t\geq0$ for all $0\leq t\leq T$.
\item[\bf{3)}] The optimal trading strategy $\hvt=(\hvt_t)_{0 \leq t \leq T}$ to the frictionless problem \eqref{P3} coincides with (the left limit of) the holdings in stock $\hvp^1_{-}=(\hvp^1_{t-})_{0 \leq t \leq T}$ of the optimal trading strategy to the utility maximisation problem \eqref{P1} under transaction costs so that $x+\hvt\sint\hS_T=V^{liq}_T(\hvp)=\hg(x)$.
\end{itemize}
\end{definition}

In Theorem 3.2 of \cite{CSY15}, we established the existence of a shadow price for a continuous price process $S=(S_t)_{0\leq t\leq T}$ satisfying the condition $(NUPBR)$ of ``no unbounded profit with bounded risk'' (without transaction costs). The assumption of $(NUPBR)$ implies that $S$ has to be a semimartingale. Therefore, our result does not yet apply to price processes driven by fractional Brownian motion $B^H=(B^H_t)_{0\leq t\leq T}$ such as the \emph{fractional Black-Scholes model}
 \begin{equation}  \label{fBS}
   S_t = \exp\big(\mu t + \sigma B^H_t\big), \qquad 0\leq t\leq T, 
 \end{equation}
 where $\mu\in\R$, $\sigma>0$ and $H\in(0,1)\backslash\{\frac{1}{2}\}$ denotes the Hurst parameter of the fractional Brownian motion $B^H$. In the present article, we combine a recent result of Peyre \cite{Pey15} with a strengthening of our existence result in Theorem 3.2 of \cite{CSY15} to fill this gap.
 
 For this, we need a weaker no arbitrage type condition than $(NUPBR)$ that is nevertheless in some sense stronger than the stickiness. It turns out that the condition $(TWC)$ of ``two way crossing'' is the suitable one to work with.
\begin{definition}\label{def:twc}
 Let $X=(X_t)_{0\leq t\leq T}$ be a real-valued continuous stochastic process and $\sigma$ a finite stopping time. 
Set
  \begin{equation*} 
   \begin{aligned}
    \sigma_+&:=\inf\{t>\sigma \,|\,X_t-X_\sigma >0\}, \\
    \sigma_-&:=\inf\{t>\sigma \,|\,X_t-X_\sigma <0\}.    
   \end{aligned}
  \end{equation*}
Then, we say that $X$ satisfies the condition $(TWC)$ of ``two way crossing'', if $\sigma_+ = \sigma_-$ $\p$-a.s. for any finite stopping time $\sigma$.
\end{definition}
The two way crossing condition was introduced by Bender in \cite{Ben12} in the analysis of the condition of ``no simple arbitrage'' (without transaction costs), that is, no arbitrage with linear combinations of buy and hold strategies. Using it in the context of portfolio optimisation under transaction costs allows us to establish the following results. For better readability, their proofs are deferred to Section \ref{sec:pr}.

\begin{theorem}\label{t1}
Fix transaction costs $\lambda\in(0,1)$ and a strictly positive continuous process $S=(S_t)_{0\leq t\leq T}$ satisfying $(TWC)$.
Let $U:(0,\infty)\to\R$ be a strictly concave, increasing, continuously differentiable utility function, satisfying the Inada conditions \eqref{Inada} and having reasonable asymptotic elasticity $AE(U):= \limsup_{x\to\infty}\frac{xU'(x)}{U(x)}<1$ and suppose that 
   \begin{equation}  \label{u(x)}
     u(x):=\sup_{\varphi\in\cA(x)}\E\big[U\big(V^{liq}_T(\varphi)\big)\big]<\infty
   \end{equation}
  for some $x>0$.
     
  Then, there exists an optimal trading strategy $\hvp=(\hvp^0_t,\hvp^1_t)_{0_-\leq t\leq T}$  for \eqref{P1} and a shadow price $\widehat{S}=(\widehat{S}_t)_{0\leq t\leq T}$.  
\end{theorem}
The significance of the condition $(TWC)$ in the above result is that it holds for the fractional Black-Scholes model \eqref{fBS} and does not require that $S$ is a semimartingale. It allows us to conclude the existence of a shadow price process for the fractional Black-Scholes model and utility functions that are bounded from above, like power utility $U(x)=\frac{x^\alpha}{\alpha}$ with risk aversion parameter $\alpha<0$. For utility functions $U:(0,\infty)\to\R$ that are not bounded from above like logarithmic utility $U(x)=\log(x)$ or power utility $ U(x)=\frac{x^\alpha}{\alpha}$ with risk aversion parameter $\alpha\in(0,1)$, it remains to show that the indirect utility \eqref{u(x)} is finite in order to apply Theorem \ref{t1}. We do this below by controlling the number of fluctuations of fractional Brownian motion of size $\delta>0$, which allows to obtain the following complete answer to the question whether or not there exists a shadow price for the fractional Black-Scholes model.
\begin{theorem}\label{t2}
  Let $U:(0,\infty)\to\R$ be a strictly concave, increasing, continuously differentiable utility function, satisfying the Inada conditions \eqref{Inada} and having reasonable asymptotic elasticity $AE(U):= \limsup_{x\to\infty}\frac{xU'(x)}{U(x)}<1$. Fix transaction costs $\lambda\in(0,1)$ and the fractional Black-Scholes model \eqref{fBS}. 
 
 Then, 
\begin{equation}  \label{uffbm}
     u(x)=\sup_{\varphi\in\cA(x)}\E\big[U\big(V^{liq}_T(\varphi)\big)\big]<\infty
   \end{equation}
for all $x>0$. In particular, there exists an optimal trading strategy $\hvp=(\hvp^0_t,\hvp^1_t)_{0_-\leq t\leq T}$ for \eqref{P1} and a shadow price $\widehat{S}=(\widehat{S}_t)_{0\leq t\leq T}$.  
\end{theorem}

As explained in Section 5 of \cite{CS15}, there is a filtered probability space $(\Om,\F,(\F_t)_{0\leq t\leq T},\p)$ supporting a Brownian motion $W=(W_t)_{0\leq t\leq T}$ that has the predictable representation property conditional on $\F_0$. The connection to frictionless financial markets then allows us to establish the following result.

\begin{theorem}\label{t3}
  Let $\big(\Om,\F,(\F_t)_{0\leq t\leq T},\p\big)$ be a filtered probability space supporting a Brownian motion $W=(W_t)_{0\leq t\leq T}$ that has the predictable representation property conditional on $\F_0$.  Fix transaction costs $\lambda\in(0,1)$ and the fractional Black-Scholes model \eqref{fBS}.
  
  Then, there exists an optimal trading strategy $\hvp=(\hvp^0_t,\hvp^1_t)_{0_-\leq t\leq T}$ and a shadow price $\widehat{S}=(\widehat{S}_t)_{0\leq t\leq T}$ for the problem of maximising logarithmic utility
\begin{equation}  \label{P1:log}
\E\big[\log\big(V^{liq}_T(\varphi)\big)\big]\to \max!, \quad \varphi\in\cA(x),
   \end{equation}
for all $x>0$.

The shadow price $\widehat{S}=(\widehat{S}_t)_{0\leq t\leq T}$ is given by an It\^o process
\begin{align}\label{t3:ito}
d \widehat{S}_t = \widehat{S}_t (\widehat{\mu}_t dt + \widehat{\sigma}_t dW_t), \quad 0 \leq t \leq T,
\end{align}
where $\widehat{\mu}=(\widehat{\mu}_t)_{0 \leq t \leq T}$ and $\widehat{\sigma}=(\widehat{\sigma}_t)_{0 \leq t \leq T}$ are predictable processes such that the solution to \eqref{t3:ito} is well-defined in the sense of Itô integration.

The coefficients $\widehat{\mu}=(\widehat{\mu}_t)_{0 \leq t \leq T}$ and $\widehat{\sigma} = (\widehat{\sigma}_t)_{0 \leq t \leq T}$ of the It\^o process \eqref{t3:ito} and the optimal trading strategy $\hvp=(\hvp^0_t,\hvp^1_t)_{0_-\leq t\leq T}$ for \eqref{P1:log} are related via
\be
\widehat{\pi}_t=\frac{\widehat{\mu}_t}{\widehat{\sigma}^2_t}=\frac{\hvp^1_{t-}\hS_t}{\hvp^0_{t-}+\hvp^1_{t-}\hS_t},\qquad 0\leq t\leq T.\label{t3:opt}
\ee
\end{theorem}

It is well known that the shadow price is related to the solution of a suitable dual problem of the primal problem \eqref{P1}; see Proposition 3.9 of \cite{CS14} for example. In the present setting, we explain how to setup this dual problem in the next section.

\section{Duality theory}\label{sec:dt}

In this section, we discuss the formulation of the dual problem of the utility maximisation problem \eqref{P1}. To that end, we recall the following notions.

 A \emph{$\lambda$-consistent price system} is a pair of stochastic processes $Z=(Z^0_t, Z^1_t)_{0 \leq t \leq T}$ consisting of the density process $Z^0=(Z^0_t)_{0 \leq t \leq T}$ of an equivalent local martingale measure $\Q\sim \p$ for a price process $\widetilde{S}=(\widetilde{S}_t)_{0 \leq t \leq T}$ evolving in the bid-ask spread $[(1-\lambda)S,S]$ and the product $Z^1=Z^0\widetilde{S}$. Requiring that $\widetilde{S}$ is a local martingale under $\Q$ is tantamount to the product $Z^1=Z^0\widetilde{S}$ being a local martingale under $\p$. Under transaction costs, $\lambda$-consistent price systems ensure ``absence of arbitrage'' in the sense of ``no free lunch with vanishing risk'' $(NFLVR)$ similarly as equivalent local martingale measures in the frictionless case. In the context of portfolio optimisation, usually not the full strength of the condition $(NFLVR)$ is needed and it is enough to have this property locally. For portfolio optimisation under transaction costs, this is captured by the notion of a $\lambda$-consistent local martingale deflator. A \emph{$\lambda$-consistent local martingale deflator} is a pair of strictly positive local martingales $Z=(Z^0_t, Z^1_t)_{0 \leq t \leq T}$ such that $\widetilde{S}:=\frac{Z^1}{Z^0}$ is evolving within the bid-ask spread $[(1-\lambda)S,S]$ and $\E[Z^0_0]=1$. We denote the set of all $\lambda$-consistent local martingale deflators by $\cZ$. Note that, if $(\tau_n)_{n=1}^\infty$ is a localising sequence of stopping times such that the stopped process $(Z^0)^{\tau_n}=(Z^0_{\tau_n\wedge t})_{0\leq t\leq T}$ is a true martingale, then $Z^{\tau_n}=(Z^0_{\tau_n\wedge t},Z^1_{\tau_n\wedge t})_{0\leq t\leq T}$ is a $\lambda$-consistent price system for the stopped process $S^{\tau_n}=(S_{\tau_n\wedge t})_{0\leq t\leq T}$. In this sense, the condition that $S$ admits a $\lambda$-consistent local martingale deflator is indeed the local version of the condition that $S$ admits a $\lambda$-consistent price system. The set $\mathcal{B}(y)$ of all \emph{$\lambda$-consistent supermartingale deflators} consists of all pairs of non-negative c\`adl\`ag supermartingales $Y=(Y^0_t, Y^1_t)_{0 \leq t \leq T}$ such that $\E[Y^0_0]=y,$ $Y^1=Y^0 \widetilde{S}$ for some $[(1-\lambda)S,S]$-valued process $\widetilde{S}=(\widetilde{S}_t)_{0 \leq t \leq T}$ and $Y^0(\varphi^0 + \varphi^1 \widetilde{S})= Y^0 \varphi^0 + Y^1 \varphi^1$ is a non-negative c\`adl\`ag supermartingale for all $\varphi \in \mathcal{A} (1).$ Note that $y\cZ\subseteq \cB(y)$ for $y>0$ by Proposition 2.6 of \cite{CSY15}. We set $\cD(y):=\{Y^0_T~|~Y=(Y^0,Y^1)\in\cB(y)\}$. By Proposition 2.9 of \cite{CSY15}, we have that $\cD(y)$ coincides with the closed, convex, solid hull of $D(y)=\{yZ^0_T~|~Z=(Z^0,Z^1)\in\cZ\}$. 
 
The following result shows how the solution to the utility maximisation problem \eqref{P1} is related to the solution of a suitable dual problem. For a continuous price process $S=(S_t)_{0\le t\le T}$, it has been established in Theorem 2.10 of \cite{CSY15}.

\begin{theorem}\label{dt}  Let $S=(S_t)_{0\le t\le T}$ be a strictly positive, continuous process. 
Suppose that $S$ admits a $\mu$-consistent local martingale deflator for all $\mu\in(0,\lambda)$, the asymptotic elasticity of $U$ is strictly less than one, i.e., $AE(U):=\limsup\limits_{x\to\infty}\frac{xU'(x)}{U(x)}<1$, and the maximal expected utility is finite,
$$u(x):=\sup_{g\in\cC(x)}\E[U(g)]<\infty,$$ for some $x\in(0,\infty)$. Then:
\bi
\item[{\bf 1)}] The primal value function $u$ and the dual value function
$$v(y):=\inf_{h\in\cD(y)}\E[V(h)],$$
where $V(y)=\sup_{x>0}\{U(x)-xy\}$ for $y>0$ denotes the Legendre transform of $-U(-x)$, are conjugate, i.e.,
\begin{eqnarray*}u(x)=\inf_{y>0}\{v(y)+xy\},\qquad v(y)=\sup_{x>0}\{u(x)-xy\},
\end{eqnarray*}
and continuously differentiable on $(0,\infty)$. The functions $u$ and $-v$ are strictly concave, strictly increasing, and satisfy the Inada conditions
$$\text{$\Lim_{x\to0}u'(x)=\infty,\qquad\Lim_{y\to\infty}v'(y)=0,\qquad\Lim_{x\to\infty}u'(x)=0,\qquad\Lim_{y\to0}v'(y)=-\infty$}.$$
\item[{\bf 2)}] For all $x,y>0$, the solutions $\widehat g (x)\in\cC(x)$ and $\widehat h(y)\in\cD(y)$ to the primal problem
\begin{equation*}
\textstyle
\E\left[U(g)\right]\to\max!, \qquad{g\in\cC(x)},
\end{equation*}
and the dual problem
\begin{equation}
\textstyle
\E\left[V(h)\right]\to\min!\label{D1}, \qquad{h\in\cD(y)},
\end{equation}
exist, are unique, and there are
$\big(\hvp^0(x),\hvp^1(x)\big)\in\cA(x)$ and $\big(\widehat{Y}^0(y),\widehat{Y}^1(y)\big)\in\cB(y)$ such that
\be
\text{$V^{liq}_T(\hvp)=\hvp^0_T(x)=\widehat g(x)\qquad$ and $\qquad\widehat{Y}^0_T(y)=\widehat h(y)$.}\label{martcond}
\ee
\item[{\bf 3)}] For all $x>0$, let $\widehat y (x)=u'(x)>0$ which is the unique solution to
$$
v(y)+xy\to\min!,\qquad y>0.
$$
Then, $\widehat g (x)$ and $\widehat h \big(\widehat y(x)\big)$ are given by $(U')^{-1}\big(\widehat h \big(\widehat y(x)\big)\big)$ and $U'\big(\widehat g (x)\big)$, respectively, and we have that $\E\big[\widehat g(x)\widehat h\big(\widehat y (x)\big)\big]=x\widehat y(x)$. In particular, the process
$$\widehat{Y}^0\big(\widehat y(x)\big)\hvp^0(x)+ \widehat{Y}^1\big(\widehat y(x)\big)\hvp^1(x)=\Big(\widehat{Y}^0_t\big(\widehat y(x)\big)\hvp^0_t(x)+ \widehat{Y}^1_t\big(\widehat y(x)\big)\hvp_t^1(x)\Big)_{0\leq t\leq T}$$
is a martingale for all $\big(\hvp^0(x),\hvp^1(x)\big)\in\cA(x)$ and $\big(\widehat{Y}^0\big(\widehat y(x)\big),\widehat{Y}^1\big(\widehat y(x)\big)\big)\in\cB\big(\widehat y (x)\big)$ satisfying \eqref{martcond} with $y=\widehat y (x)$.
\item[{\bf 4)}] 
Moreover, for $\hS=\tfrac{\widehat{Y}^1}{\widehat{Y}^0}$, we have
$$\widehat{Y}^0\big(\widehat y(x)\big)\hvp^0(x)+ \widehat{Y}^1\big(\widehat y(x)\big)\hvp^1(x)=\widehat{Y}^0\big(\widehat y(x)\big)\left(x+ \hvp^1_{-}(x)\sint \hS\right).$$
This implies in particular that 
\begin{align*}
\{d\hvp^{1,c}>0\}&\subseteq \{\hS=S\}, & \{d\hvp^{1,c}<0\}&\subseteq \{\hS=(1-\lambda)S\},\notag \\
\{\Delta \hvp^1>0\}&\subseteq \{\hS=S\}, &  \{\Delta \hvp^1<0\}&\subseteq \{\hS=(1-\lambda)S	\}. \notag
\end{align*}
\item[{\bf 5)}] Finally, we have $v(y)=\Inf_{(Z^0,Z^1)\in\mathcal{Z}}\E[V(yZ^0_T)]$.
\ei
\end{theorem}
\bp
See Theorem 2.10 of \cite{CSY15} and Remark 2.13 of \cite{CSY15}. Compare also Theorems 3.2 and Theorem 3.5 of \cite{CS14}.
\ep

In order to apply the theorem above to prove Theorem \ref{t1}, we need to show that the condition $(TWC)$ of ``two way crossing'' implies the existence of $\mu$-consistent local martingale deflators for all $\mu\in(0,1)$. This follows from a local version of the fundamental theorem of asset pricing for continuous processes under small transaction costs. For this, we use the subsequent no arbitrage concepts; compare \cite{GRS10}. 

\begin{definition}\label{def:NOIA}
 Let $S=(S_t)_{0\le t\le T}$ be a strictly positive, continuous process. 
 We say that $S$ allows for an ``obvious arbitrage'', if there are $\alpha >0$ and $[0,T]\cup\{\infty\}$-valued stopping times $\sigma\leq\tau$ with 
   $\p[\sigma <\infty]=\p[\tau <\infty] >0$ such that either
    $$ (a)\quad S_\tau \geq (1+\alpha) S_\sigma, \quad \ \mbox{a.s.~on} \ \{\sigma <\i\},$$
  or
    $$ (b)\quad S_\tau \leq \frac{1}{1+\alpha} S_\sigma, \quad\quad \mbox{a.s.~on} \ \{\sigma <\i\}.$$
 In the case of (b) we also assume that $(S_t)_{\sigma \leq t \leq \tau}$ is uniformly bounded.

 We say that $S$ allows for an ``obvious immediate arbitrage'', if, in addition, we have 
   $$ (a)\quad S_t\geq S_\sigma, \quad \mbox{ for }  \sigma\leq t\leq \tau,  \, \mbox{a.s.~on} \ \{\sigma <\i\}, $$
  or
   $$ (b)\quad S_t\leq S_\sigma, \quad \mbox{ for }  \sigma\leq t\leq \tau,  \, \mbox{a.s.~on} \ \{\sigma <\i\}. $$

 We say that $S$ satisfies the condition $(NOA)$ (respectively, $(NOIA)$) of ``no obvious arbitrage'' (respectively, ``no obvious immediate arbitrage''), 
   if no such opportunity exists.
\end{definition}

It is indeed rather obvious how to make an arbitrage if $(NOA)$ fails, provided the transaction costs $0 <\lambda <1$ are smaller than $\alpha$.
Assuming, e.g., condition $(a)$, one goes long in the asset $S$ at time $\sigma$ and closes the position at time $\tau$. 
In case of an obvious immediate arbitrage one is in addition assured that during such an operation the stock price will never fall under the initial value $S_\sigma$. 
In particular this gives an unbounded profit with bounded risk under transaction costs $\lambda$.  

In the case of condition $(b)$, one does a similar operation by going short in the asset $S$. 
The boundedness condition in the case (b) of $(NOA)$ makes sure that this strategy is admissible.

Using $(NOIA)$ in addition to $(NOA)$ then allows us to obtain the following local version of the fundamental theorem of asset pricing for continuous processes under small transaction costs, which is a slight strengthening of Theorem 1 of \cite{GRS10}. 

\begin{theorem} \label{localFTAP}
  Let $S=(S_t)_{0\leq t\leq T}$ be a strictly positive, continuous process. The following assertions are equivalent.
  \begin{enumerate}[$(i)$]
   \item Locally, there is no obvious immediate arbitrage $(NOIA)$. 
   \item Locally, there is no obvious arbitrage $(NOA)$.
   \item Locally, for each $0<\mu<1$, there exists a $\mu$-consistent price system.
   \item For each $0<\mu<1$, there exists a $\mu$-consistent local martingale deflator.
  \end{enumerate}
\end{theorem}

\begin{proof}
 Obviously, we have $(ii) \Rightarrow(i)$. The equivalent $(ii)\Leftrightarrow(iii)$ follows directly from Theorem 1 of \cite{GRS10}. As explained above, $(iv)$ implies $(iii)$. 
 
 \vspace{3mm}
 
 The converse $(iii)\Rightarrow(iv)$ follows by exploiting that $(iii)$ asserts locally the existence of a $\mu$-consistent price system for \emph{each} $0<\mu<1$. Indeed, fix $0<\mu<1$ and a localising sequence $(\tau_n)_{n=1}^\infty$ of stopping times. Let $\overline{Z}=(\overline{Z}^0_{\tau_n\wedge t},\overline{Z}^1_{\tau_n\wedge t})_{0 \leq t \leq T}$ be a $\overline{\mu}$-consistent price system for $S^{\tau_n}=(S_{\tau_n\wedge t})_{0\leq t\leq T}$ with $0<\overline{\mu}<\mu$. Then we can extend $\overline{Z}$ to a $\widetilde{\mu}$-consistent price system $\widetilde{Z}=(\widetilde{Z}^0_{\tau_{n+1}\wedge t},\widetilde{Z}^1_{\tau_{n+1}\wedge t})_{0 \leq t \leq T}$ for $S^{\tau_{n+1}}=(S_{\tau_{n+1}\wedge t})_{0\leq t\leq T}$ with $0<\overline{\mu}<\widetilde{\mu}<\mu$ by setting 
\begin{equation*}
\widetilde{Z}^0_t= 
\begin{cases}
\overline{Z}^0_t&: 0 \leq t < \tau_{n},\\
\check{Z}^0_{\tau_{n+1}\wedge t} \frac{\overline{Z}^0_{\tau_{n}}}{\check{Z}^0_{\tau_{n}} }&: \tau_{n} \leq t \leq T,  \\
\end{cases}
\end{equation*}
\begin{equation*}
\widetilde{Z}^1_t= 
\begin{cases}
(1-\check\mu)\overline{Z}^1_t&: 0 \leq t < \tau_{n},\\
(1-\check\mu)\check{Z}^1_{\tau_{n+1}\wedge t} \frac{\overline{Z}^1_{\tau_{n}}}{\check{Z}^1_{\tau_{n}} }&: \tau_{n} \leq t \leq T,  \\
\end{cases}
\end{equation*}
where $\check{Z}=(\check{Z}^0_{\tau_{n+1}\wedge t},\check{Z}^1_{\tau_{n+1}\wedge t})_{0 \leq t \leq T}$ is a $\check{\mu}$-consistent price system for $S^{\tau_{n+1}}=(S_{\tau_{n+1}\wedge t})_{0\leq t\leq T}$ with $0<\check{\mu}<\frac{\widetilde{\mu}-\overline{\mu}}{2}$. Repeating this extension allows us to establish the existence of a $\mu$-consistent local martingale deflator. 
 
 \vspace{3mm}
 
 \noindent $(i)\Rightarrow(iii)$: As $(iii)$ is a local property, we may assume that $S$ satisfies $(NOIA)$. 
 
 To prove $(iii)$, we do a similar construction as in the proof of Proposition 1 in \cite{GRS10}. 
 We suppose in the sequel that the reader is familiar with the aforementioned proof. 
 
 Define the stopping time $\bar{\varrho}_1$ by
  $$ \bar{\varrho}_1 :=\inf\left\{t>0\,\Bigg|\,\frac{S_t}{S_0} \geq 1+\mu \,\mbox{ or }\, \frac{S_t}{S_0} \leq \frac{1}{1+\mu}\right\}. $$
 Define the sets $\overline{A}_1^+$, $\overline{A}_1^-$ and $\overline{A}_1^0$ as 
   \begin{equation*}
    \begin{aligned}
      \overline{A}_1^+ &:=\big\{\bar{\varrho}_1<\infty, \, S_{\bar{\varrho}_1}=(1+\mu)S_0\big\}, \\
      \overline{A}_1^- &:=\left\{\bar{\varrho}_1<\infty, \, S_{\bar{\varrho}_1}=\frac{S_0}{(1+\mu)}\right\}, \\
      \overline{A}_1^0 &:=\big\{\bar{\varrho}_1=\infty\big\}.
    \end{aligned}
   \end{equation*}
 
 It was observed in \cite{GRS10} that the assumption $(NOA)$ rules out the case $\p\big[\overline{A}_1^+\big]=1$ and $\p\big[\overline{A}_1^-\big]=1$. 
 But under the present weaker assumption $(NOIA)$ we cannot a priori exclude the above possibilities. 
 To refine the argument from \cite{GRS10} in order to apply to the present setting, we distinguish two cases. 
 Either we have $\p\big[\overline{A}_1^+\big]<1$ and $\p\big[\overline{A}_1^-\big]<1$, or one of the probabilities $\p\big[\overline{A}_1^+\big]$ or $\p\big[\overline{A}_1^-\big]$ equals one. 
 
 In the first case, we let $\varrho_1:=\bar{\varrho}_1$ and proceed exactly as in the proof of Proposition 1 in \cite{GRS10} to complete the first inductive step. 
 
 For the second case, we assume w.l.o.g.~that $\p\big[\overline{A}_1^+\big]=1$, the other case can be treated in an analogous way. 
 
 Define the real number $\beta \leq 1$ as the essential infimum of the random variable $\min_{0\leq t\leq\bar{\varrho}_1} \frac{S_t}{S_0}$.
 We must have $\beta <1$, otherwise the pair $(0,\bar{\varrho}_1)$ would define an immediate obvious arbitrage. 
 We also have the obvious inequality $\beta \geq \frac{1}{1+\mu}.$
 
 We define for $1>\gamma\geq\beta$ the stopping time 
  $$ \bar{\varrho}^\gamma_1:=\inf\left\{t>0 \, \Bigg| \, \frac{S_t}{S_0} \geq 1+\mu \, \mbox{ or } \, \frac{S_t}{S_0} \leq \gamma\right\}. $$
 Defining 
  \begin{equation*}
   \begin{aligned}
    \overline{A}^{\gamma,+}_1 := \big\{S_{\bar{\varrho}_1^\gamma}=(1+\mu)S_0 \big\} \quad \mbox{ and } \quad
    \overline{A}^{\gamma,-}_1 := \left\{S_{\bar{\varrho}_1^\gamma} = \gamma S_0\right\},
   \end{aligned}
  \end{equation*}
  we find an almost surely partition of $\overline{A}^+_1$ into the sets $\overline{A}^{\gamma,+}_1$ and $\overline{A}^{\gamma,-}_1$. 
 Clearly $\p[\overline{A}^{\gamma,-}_1]>0$, for $1>\gamma >\beta$. 
 We claim that $$ \lim\limits_{\gamma\searrow\beta}\p\big[\overline{A}^{\gamma,-}_1\big]=0. $$
 Indeed, supposing that this limit were positive, we again could find an obvious immediate arbitrage, as in this case we have that $\p\big[\overline{A}^{\beta,-}_1\big]>0$. 
 Hence, the pair of stopping times
  \begin{equation*}
    \begin{aligned}
      \sigma &= \bar{\varrho}^\beta_1\mathbbm{1}_{\big\{S_{\bar{\varrho}^\beta_1}=\beta S_0\big\}} +\infty\mathbbm{1}_{\big\{S_{\bar{\varrho}^\beta_1}=(1+\mu)S_0\big\}} \\
      \tau &= \bar{\varrho}_1\mathbbm{1}_{\big\{S_{\bar{\varrho}^\beta_1}=\beta S_0\big\}} +\infty\mathbbm{1}_{\big\{S_{\bar{\varrho}^\beta_1}=(1+\mu)S_0\big\}}
    \end{aligned}
  \end{equation*}
  would define an obvious immediate arbitrage, which is contrary to our assumption. 
  
 We thus may find $1>\gamma >\beta$ such that
\be
0<\p\big[\overline{A}^{\gamma,-}_1\big]<\frac{1}{2}.\label{W1}
\ee
 After having found this value of $\gamma$ we can define the stopping time $\varrho_1$ in its final form as
   $$ \varrho_1:=\bar{\varrho}_1^\gamma. $$
 
 Next, we define the sets
  \begin{equation} \label{W2}
   \begin{aligned}
     & A^+_1 := \{\varrho_1 <\i,\, S_{\varrho_1} =(1+\mu) S_0\}=\overline{A}^{\gamma,+}_1  \\
     & A^-_1 := \{\varrho_1 <\i,\, S_{\varrho_1} =\gamma S_0\}=\overline{A}^{\gamma,-}_1
   \end{aligned}
  \end{equation}
  to obtain a partition of $\Omega$ into two sets of positive measure. 

 As in the proof of Proposition 1 in \cite{GRS10}, we define a probability measure $\Q_1$ on $\F_{\varrho_1}$ by letting $\frac{d\Q_1}{d\p}$ be constant on these two sets, 
  where the constants are chosen such that
    $$ \Q_1[A^+_1] = \frac{1-\gamma}{1+\mu-\gamma} \quad \mbox{ and } \quad \Q_1[A^-_1] = \frac{\mu}{1+\mu-\gamma}. $$ 
 We then may define the $\Q_1$-martingale $(\widetilde{S}_t)_{0 \leq t \leq \varrho_1}$ by 
    $$ \widetilde{S}_t := E_{\Q_1} [S_{\varrho_1}| \F_t], \quad \quad 0\leq t \leq\varrho_1, $$
   to obtain a process remaining in the interval $[\gamma S_0, (1+\mu)S_0].$ 
   
 The above weights for $\Q_1$ were chosen in such a way to obtain
    $$ \widetilde{S}_0 = E_{\Q_1}[S_{\varrho_1}]=S_0. $$
    
 This completes the first inductive step similarly as in the proof of Proposition 1 of \cite{GRS10}. 
 Summing up, we obtained $\varrho_1$, $\Q_1$ and $(\widetilde{S}_t)_{0\leq t\leq\varrho_1}$ precisely as in the proof of Proposition 1 in \cite{GRS10} with the following additional possibility: it may happen that $\varrho_1$ does not stop when $S_t$ first hits $(1+\mu)S_0$ or $\frac{S_0}{1+\mu}$, but rather when $S_t$ first hits $(1+\mu)S_0$ or $\gamma S_0$, for some $\frac{1}{1+\mu} < \gamma < 1.$ In the case we have $\p[A^0_1]=0$, we made sure that $\p[A^-_1] < \frac{1}{2}$, i.e., we have a control on the probability of $\{S_{\varrho_1}=\gamma S_0\}.$
 
 \vskip10pt
 
 We now proceed as in the proof of Proposition 1 in \cite{GRS10} with the inductive construction of $\varrho_n, \Q_n$ and $\big(\widetilde{S}_t\big)_{0\leq t\leq\varrho_n}$. 
 The new ingredient is that again we have to take care (conditionally on $\F_{\varrho_{n-1}}$) of the additional possibility $\p[A^+_n]=1$ or $\p[A^-_n]=1$. 
 Supposing again w.l.o.g.~that we have the first case, we deal with this possibility precisely as for $n=1$ above, but now we make sure that $\p[A^-_n]<2^{-n}$ instead of $\p[A^-_1]=\p\big[\overline{A}^{\gamma,-}_1\big]<\frac{1}{2}$ as in \eqref{W1} and \eqref{W2} above. 
 
 This completes the inductive step and we obtain, for each $n\in \mathbb{N}$, an equivalent probability measure $\Q_n$ on $\F_{\varrho_n}$ and a $\Q_n$-martingale $\big(\widetilde{S}_t\big)_{0 \leq t \leq \varrho_n}$ taking values in the bid-ask spread $([\frac{1}{1+\mu}S_t,(1+\mu)S_t])_{0 \leq t \leq \varrho_n}.$ 
 We note in passing that there is no loss of generality in having chosen this normalization of the bid-ask spread instead of the usual normalization $[(1-\mu')S', S']$ by passing from $S$ to $S'=(1+\mu)S$ and from $\mu$ to $\mu'=1-\frac{1}{(1+\mu)^2}.$ 
 
 There is one more thing to check to complete the proof of $(iii):$ we have to show that the stopping times $(\varrho_n)^\infty_{n=1}$ increase almost surely to infinity. 
 This is verified in the following way: suppose that $(\varrho_n)^\infty_{n=1}$ remains bounded on a set of positive probability. 
 On this set we must have that $\frac{S_{\varrho_{n+1}}}{S_{\varrho_n}}$ equals $(1+\mu)$ or $\frac{1}{1+\mu}$, except for possibly finitely many $n's$. 
 Indeed, the above requirement $\p[A^-_n]< 2^{-n}$ makes sure that a.s.~the novel possibility of moving by a value different from $(1+\mu)$ or $\frac{1}{1+\mu}$ can only happen finitely many times. 
 Therefore we may, as before Proposition 1 in \cite{GRS10}, conclude from the uniform continuity and strict positivity of the trajectories of $S$ on $[0,T]$ that $\varrho_n$ increases a.s.~to infinity which completes the proof of $(iii)$. 
\end{proof}

\section{Fluctuations of fractional Brownian motion}\label{sec:fl}

In this section, we establish a control on the number of fluctuations of fractional Brownian motion. It allows us to establish the finiteness of the indirect utility for the proof of Theorem \ref{t2} but might also be interesting in its own right.

Let $B^H=(B^{H}_t)_{t \geq 0}$ be a standard fractional Brownian motion with Hurst parameter $H \in (0, 1]$. Fix $\delta > 0$ and define the \emph{$\delta$-fluctuation times} $(\tau_j)_{j \geq 0}$ of $B^H$ inductively by $\tau_0 \equiv 0$ and
\[
\tau_{j + 1} (\omega) \coloneqq \inf \big\{t \geq \tau_j(\omega)~\big|~ \abs{ B^H_t(\omega)-B^H_{\tau_j}(\omega)} \geq \delta\big\}
.\]
The \emph{number of $\delta$-fluctuations up to time  $T \in (0, \infty)$} is the random variable
\[
F^{(\delta)}_T (\omega) \coloneqq \sup \{j \geq 0~|~\tau_j (\omega) \leq T\} \label{dfluB}
.\]
The goal of this section is to prove the following proposition.
\begin{proposition}\label{ppn:fluctuations}
With the notation above, there exist finite positive constants $C = C (H),\allowbreak C' = C' (H)$ only depending on $H$ such that
\[
\p \big[F^{(\delta)}_T \geq n\big] \leq C' \exp \big(-C^{-1} \delta^2 T^{-2 H} n^{1 + (2 H \wedge 1)}\big)\quad\text{for all $ n \in \N$}\label{fluc}
.\]
\end{proposition}

The interest of Proposition~\ref{ppn:fluctuations} for this article lies in the following corollary.
\begin{corollary}\label{cor:fluc}
With the above notation, the random variable $F^{(\delta)}_T$ does have exponential moments of all orders, that is,
\[
\E \brack[\big]{\exp \big(a F^{(\delta)}_T\big)} < \infty\quad\text{for all $a\in\R$}\label{cor:fluc:eq1}
.\]

Moreover, if $H \geq 1/2$, this random variable even has a Gaussian moment, that is, there exists $a>0$ such that
\[
\E \brack[\big]{\exp \big(a (F^{(\delta)}_T)^2\big)} < \infty\label{cor:fluc:eq2}
.\]
\end{corollary}
\bp[\proofname\ of Corollary \ref{cor:fluc}]
For $f(x)=\exp(ax)$ and $f(x)=\exp(ax^2)$, we have
$$\E\big[f(F^{(\delta)}_T)\big]=\int_0^\infty f'(x)\p\big[F^{(\delta)}_T\geq x\big]dx$$
by Fubini's Theorem. Combining this with the estimate \eqref{fluc} gives \eqref{cor:fluc:eq1}. For  \eqref{cor:fluc:eq2}, we have to use completion of squares in addition.
\ep

\begin{proof}[\proofname\ of Proposition~\ref{ppn:fluctuations}]
Throughout the proof, we denote by $C, C'>0$ constants only depending on $H$, but whose precise value may vary from appearance to appearance. 

Let $n,m \in \N$ be such that $m> n$. We then divide $[0,T]$ into $m$ subintervals $I_k:=[\frac{k-1}{m}T,\frac{k}{m}T]$ for $k=1,\ldots,m$ and denote their midpoints by $t_k:=\frac{k-\frac{1}{2}}{m}T$.  By Fernique's Theorem (see, e.g., Lemma 4.2 of \cite{Pey15}), we can estimate the probability of the set $
A_1:=\bigcup_{k=1}^m\big\{|B^H_t-B^H_{t_k}|>\textstyle\frac{\delta}{4}\ \text{for $t\in I_k$}\big\}
$
by 
\be  \label{P(A1)estimate}
\p[A_1]\leq m \p\big[\big|B^H_t-B^H_{t_k}\big|>\textstyle\frac{\delta}{4}\ \text{for $t\in I_k$}\big]\leq m C'\exp\big(-C^{-1}\delta^2T^{-2H\wedge 1}m^{2H\wedge 1}\big).
\ee
On the complement $A^c_1$ of $A_1$, we then have that 
\be
\sup_{t\in I_k}\big|B^H_t-B^H_{t_k}\big|\leq\textstyle\frac{\delta}{4}\ \text{for all $k=1,\ldots,m$.}\label{C1}
\ee

Suppose now that $F^{(\delta)}_T(\om)\geq n$. Then there have to be at least $n+1$ ``random indices'' $0=K_0(\om)<K_1(\om)<\ldots<K_n(\om)\leq m$ such that $\tau_j(\om)\in I_{K_j(\om)}$ for $j=1,\ldots,n$. Because of \eqref{C1} and $|B^H_{\tau_j}-B^H_{\tau_{j-1}}|=\delta$, we then must have 
$$\big|B^H_{t_{K_j}}-B^H_{t_{K_{j-1}}}\big|\geq\frac{\delta}{2}$$
for $j=1,\ldots,n$ on $\{F^{(\delta)}_T\geq n\}\cap A_1^c$.

In order to estimate $\p\big[\{F^{(\delta)}_T\geq n\}\cap A_1^c\big]$, it therefore only remains to bound the probability of the event 
$$A_2:=\bigcap_{j=1}^n\big\{|B^H_{t_{K_j}}-B^H_{t_{K_{j-1}}}|\geq\textstyle\frac{\delta}{2}\big\}.$$
But the event $A_2$ depends on the realisation of the ``random indices'' $0=K_0(\om)<K_1(\om)<\ldots<K_n(\om)\leq m$. To get rid of this dependence, we simply estimate the probability of the event
$$A_3:=\bigcap_{j=1}^n\big\{|B^H_{t_{k_j}}-B^H_{t_{k_{j-1}}}|\geq\textstyle\frac{\delta}{2}\big\}$$
for all $\binom{m}{n}$ possible realisations $0=k_0<k_1<\ldots<k_n\leq m$ of the ``random indices'' $0=K_0(\om)<K_1(\om)<\ldots<K_n(\om)\leq m$.

For this, fix an arbitrary realisation of indices $0=k_0<k_1<\ldots<k_n\leq m$ and set 
$$\Delta_j:=B^H_{t_{k_j}} - B^H_{t_{k_{j - 1}}}\quad\text{for $j=1,\ldots,m.$}$$
Then
$$A_3=\bigcap_{j=1}^n\big\{|\Delta_j|\geq {\textstyle \frac{\delta}{2}}\big\} =\bigcap_{j=1}^n \big\{\sign(\Delta_j)\Delta_j\geq \textstyle\frac{\delta}{2}\big\}\subseteq\big\{\sum_{j=1}^n\sign(\Delta_j)\Delta_j\geq n\textstyle\frac{\delta}{2}\big\}$$
so that 
$$A_3\subseteq \bigcup_{j=1}^n\bigcup_{\ve_j\in\{-1,+1\}}\left\{\sum_{j=1}^n\ve_j\Delta_j\geq\textstyle n \frac{\delta}{2}\right\}.$$

For fixed $\ve_j\in\{-1,+1\}$, where $j=1,\ldots,n$, we have that $\sum_{j=1}^n\ve_j\Delta_j$ is a centred normally distributed random variable with variance
\[
\Var\left(\sum_{j=1}^n\ve_j\Delta_j\right)=\sum_{j, j' = 1}^{n} \ve_j \ve_{j'} \Cov (\Delta_j, \Delta_{j'}) \leq\sum_{j, j' = 1}^{n}\abs{\Cov (\Delta_j, \Delta_{j'})}\label{C2}
.\]

To estimate \eqref{C2}, we distinguish the following two cases:
\bi
\item[(i)] $H\geq \frac{1}{2}$.
\item[(ii)] $H< \frac{1}{2}$.
\ei

If $H \geq 1/2$, the covariance $\Cov (\Delta_j, \Delta_{j'})$ is always non-negative, so that
    \begin{equation*}
     \sum_{j, j' = 1}^{n} \abs{\Cov (\Delta_j, \Delta_{j'})} = \Var \left({\sum_{j = 1}^{n} \Delta_i}\right)
                                                      = \Var (B^H_{t_{k_{n}}} - B^H_{t_{k_{0}}}) = \abs{t_{k_n} - t_{k_0}}^{2H} \leq T^{2 H}.
    \end{equation*}

If $H < 1/2$, the covariance $\Cov (\Delta_j, \Delta_{j'})$ is non-positive as soon as $j \neq j'$, so that
    \begin{equation*}
      \begin{aligned}
        \sum_{j' = 0}^{n - 1} \abs{\Cov (\Delta_j, \Delta_{j'})} &= \Var (\Delta_j) - \Cov \paren[\Big]{\Delta_j, \sum_{j' < j} \Delta_{j'}} - \Cov \paren[\Big]{\Delta_j, \sum_{j' > j} \Delta_{j'}} \\
                                                                 &= \Var (\Delta_j) - \Cov (\Delta_j, B^H_{t_{k_{j-1}}}-B^H_{t_{k_{0}}}) - \Cov (\Delta_j, B^H_{t_{k_{n}}} - B^H_{t_{k_{j}}}). 
      \end{aligned}
    \end{equation*}
   For $0\leq t \leq u \leq v\leq T$, it follows from the definition of fractional Brownian motion that
    \begin{equation*}
      -\Cov (B_t - B_u, B_u - B_v) = \tfrac{1}{2} (\abs{t - u}^{2 H} + \abs{u - v}^{2 H} - \abs{t - v}^{2 H}) \leq \tfrac{1}{2} \abs{t - u}^{2 H}.
    \end{equation*}
    Therefore, we have that
    \begin{equation*}
      \begin{aligned}
        \sum_{j' = 1}^{n } \abs{\Cov (\Delta_j, \Delta_{j'})} &\leq \abs{t_{k_j} - t_{k_{j - 1}}}^{2 H} + \tsfrac{1}{2} \abs{t_{k_j} - t_{k_{j - 1}}}^{2 H} + \tsfrac{1}{2} \abs{t_{k_j} - t_{k_{j - 1}}}^{2 H}  \\
                                                                  &= 2\abs{t_{k_j} - t_{k_{j - 1}}}^{2 H}
      \end{aligned}
    \end{equation*}
    and thus
    \begin{equation} \label{eqn:bound-with-ti-for-small-H}
      \Var \left(\sum_{j = 1}^{n } \ve_j \Delta_j\right) \leq 2 \sum_{j = 1}^{n } \abs{t_{k_j} - t_{k_{j - 1}}}^{2 H}.
    \end{equation}
   But, since $H < 1/2$, the function $x \mapsto x^{2H}$ is concave, so that the right-hand side of \eqref{eqn:bound-with-ti-for-small-H} is bounded above by
    \begin{equation*}
     2 n \paren[\bigg]{\frac{\sum_{j = 1}^{n} \abs{t_{k_j} - t_{k_{j - 1}}}}{n}}^{2 H} = 2 n (\abs{t_{k_n} - t_{k_0}} \div n)^{2 H} \leq 2 n (T \div n)^{2H} = 2 n^{1 - 2 H} T^{2 H}.
    \end{equation*}

Hence in both case we can estimate
\[
\Var \paren[\Big]{\sum_{j = 1}^{n} \ve_j \Delta_j} \leq 2 T^{2 H} n^{(1 - 2 H)_+}.
\]
So, using the classical bound that $\p[Z \geq x] \leq e^{-x^2 / 2}$ for any standard normal distributed random variable $Z\sim \mathcal{N}(0,1)$, we have that
$$\p\left[\sum_{j=1}^n\ve_j\Delta_j\geq{\textstyle n \frac{\delta}{2}}\right]\leq\exp\left(-T^{-2 H} \delta^2 n^{1 + (2 H \wedge 1)} / 16\right) $$
for all possible $2^n$ choices of $\ve_j\in\{-1,+1\}$, where $j=1,\ldots,n$, and therefore
$$\p\big[A_3\big]\leq 2^n\exp\left(-T^{-2 H} \delta^2 n^{1 + (2 H \wedge 1)} / 16\right).$$

Combining that estimate with \eqref{P(A1)estimate}, and using that $\binom{m}{n} \leq m^{n} \div n!$, we finally get that
 \begin{equation*}
   \begin{aligned}
    \p \big[F^{(\delta)}_T \geq n\big] &\leq \p[A_1]+\p \big[\big\{F^{(\delta)}_T \geq n\big\}\cap A^c_1\big]\\
                    &\leq C' m \exp\big(-C^{-1} T^{-2 H} \delta^2 m^{2H\wedge 1}\big)+ \frac{2^n m^{n}}{n!} \exp \big(-T^{-2 H} \delta^2 n^{1 + (2 H \wedge 1)} / 16\big).
   \end{aligned}
 \end{equation*}
Now, it only remains to choose $m$ to be the smallest integer such that $m\geq n^{\frac{1}{H}}$ to obtain 
   $$\p\big[F^{(\delta)}_T \geq n\big]\leq C' \exp\left(-C^{-1} \delta^2 T^{-2 H} n^{1 + (2 H \wedge 1)}\right), $$
which completes the proof.
\end{proof}

\section{Proofs of the main results}\label{sec:pr}
\begin{proof}[Proof of Theorem \ref{t1}]
  Like in the proof of Theorem 3.2 in \cite{CSY15}, we show the existence of a shadow price by duality. 
  To that end, we observe that, for continuous price processes $S=(S_t)_{0\leq t\leq T}$, the condition $(TWC)$ of ``two way crossing'' implies the no obvious immediate arbitrage condition $(NOIA)$ locally. 
  It follows by part (iv) of Theorem \ref{localFTAP} that there exists a $\mu$-consistent local martingale deflator for $S$ for each $\mu\in(0,1)$. Therefore, the assumptions of Theorem \ref{dt} are satisfied and there exists an optimal trading strategy $\widehat{\varphi}=(\widehat{\varphi}^0_t,\widehat{\varphi}^1_t)_{0_-\leq t\leq T}$ that attains the supremum in \eqref{u(x)} as well as a supermartingale deflator $\widehat{Y}=(\widehat{Y}^0_t,\widehat{Y}^1_t)_{0\leq t\leq T}$ that solves the dual problem \eqref{D1}. 
  
  To obtain the existence of a shadow price $\widehat{S}=(\widehat{S}_t)_{0\leq t\leq T}$ for problem \eqref{u(x)} above (in the sense of Definition \ref{def:sp}), it is by Proposition 3.7 of \cite{CS14} sufficient to show that the dual optimiser $\widehat{Y}=(\widehat{Y}^0_t,\widehat{Y}^1_t)_{0\leq t\leq T}$ is a local martingale. 
  By Proposition 3.3 of \cite{CSY15}, this follows as soon as we have that the liquidation value
    \begin{equation*}
     V^{liq}_t(\hvp) := \hvp_t^0 + (\hvp_t^1)^+(1-\lambda)S_t - (\hvp^1_t)^-S_t
    \end{equation*}
   is strictly positive almost surely for all $t\in [0,T]$, i.e., 
    \begin{equation}  \label{p3}
       \inf_{0\leq t\leq T}V^{liq}_t(\widehat{\varphi}) >0, \quad a.s. 
    \end{equation}
   
  To show \eqref{p3}, we argue by contradiction. 
  Define 
    \begin{equation}  \label{N1a}
     \sigma_\varepsilon := \inf\left\{t\in[0,T]~\big|~ V^{liq}_t(\hvp)\leq \varepsilon\right\}, 
    \end{equation}
    and let $\sigma:=\lim_{\varepsilon\searrow 0}\sigma_{\varepsilon}$. 
  We have to show that $\sigma=\infty$, almost surely.     
  Suppose that $\p[\sigma<\infty]>0$ and let us work towards a contradiction. 

  First observe that $V^{liq}_\sigma(\hvp)=0$ on $\{\sigma<\infty\}$. 
  Indeed, as $\big(V^{liq}_t(\hvp)\big)_{0\leq t\leq T}$ is c\`adl\`ag, we have that 
    $$ 0\leq V^{liq}_\sigma(\hvp) \leq \lim_{\varepsilon\searrow 0}V^{liq}_{\sigma_\varepsilon}(\hvp)\leq 0 $$
   on the set $\{\sigma<\infty\}$.

  So suppose that $V^{liq}_\sigma(\hvp)=0$ on the set $\{\sigma<\infty\}$ with $\p[\sigma<\infty]>0$. 
  We may and do assume that $S$ ``moves immediately after $\sigma$'', i.e., $\sigma=\inf\{t>\sigma~|~S_t\ne S_{\sigma}\}$.
  Indeed, we may replace $\sigma$ on $\{\sigma<\infty\}$ by the stopping time $\sigma_+=\sigma_-$. 
  As $V^{liq}_T(\widehat{\varphi})>0$ a.s., we have $\sigma_+<T$ on $\{\sigma<\infty\}$. 

  We shall repeatedly use the fact established in Theorem \ref{dt} that the process 
    $$\widehat{V}=\big(\hvp^0_t\hY^0_t+\hvp^1_t\hY^1_t\big)_{0\leq t\leq T}$$
   is a uniformly integrable $\p$-martingale satisfying $\widehat{V}_T>0$ almost surely.

This implies that $\hvp^1_\sigma\ne 0$ a.s.~on $\{\sigma<\infty\}$. 
  Indeed, otherwise $\widehat{V}_\sigma=\widehat{Y}^0_\sigma V^{liq}_\sigma(\hvp)= 0$ on $\{\sigma<\infty\}$. 
  As $\widehat{V}$ is a uniformly integrable martingale with strictly positive terminal value $\widehat{V}_T>0$, we arrive at the desired contradiction.

  We consider here only the case that $\hvp^1_\sigma>0$ on $\{\sigma<\infty\}$ almost surely. 
  The case $\hvp^1_\sigma<0$ with strictly positive probability on $\{\sigma<\infty\}$ can be dealt with in an analogous way. 
  Next, we show that we cannot have $\hS_\sigma=(1-\lambda)S_\sigma$ with strictly positive probability on $\{\sigma<\infty\}$. 
  Indeed, this again would imply that $\widehat{V}_\sigma=\widehat{Y}^0_\sigma V^{liq}_\sigma(\hvp)= 0$ on this set which yields a contradiction as in the previous paragraph.

  Hence, we assume that $\hS_\sigma>(1-\lambda)S_\sigma$ on $\{\sigma<\infty\}$. 
  This implies by in part 4) of Theorem \ref{dt} that the utility-optimising agent defined by $\hvp$ cannot sell stock at time $\sigma$ as well as for some time after $\sigma$, as $S$ is continuous and $\widehat{S}$ c\`adl\`ag. 
  Note, however, that the optimising agent may very well buy stock. 
  But, we shall see that this is not to her advantage. 

  Define the stopping time $\varrho_n$ as the first time after $\sigma$ when one of the following events happens
    \begin{enumerate} [(i)]
      \item $\hS_t-(1-\lambda)S_t< \frac{1}{2}\big(\hS_\sigma-(1-\lambda)S_\sigma\big)$ or
      \item $S_t<S_\sigma-\frac{1}{n}$.
    \end{enumerate}
  By the hypothesis of $(TWC)$ of ``two way crossing'', we conclude that, a.s.~on $\{\sigma<\infty\}$, we have that $\varrho_n$ decreases to $\sigma$ and that 
    we have $S_{\varrho_n}=S_\sigma-\frac{1}{n}$, for $n$ large enough. 
  Choose $n$ large enough such that $S_{\varrho_n}=S_\sigma-\frac{1}{n}$ on a subset of $\{\sigma<\infty\}$ of positive measure. 
  Then $V^{liq}_{\varrho_n}(\hvp)$ is strictly negative on this set which will give the desired contradiction. 
  Indeed, the assumption $\hvp^1_\sigma>0$ implies that the agent suffers a strict loss from this position as $S_{\varrho_n}<S_\sigma$. 
  The condition (i) makes sure that the agent cannot have sold stock between $\sigma$ and $\varrho_n$.  
  The agent may have bought additional stock during the interval $\llbracket \sigma, \varrho_n\rrbracket$. 
  However, this cannot result in a positive effect either as the subsequent calculation, which holds true on $\{S_{\varrho_n}=S_\sigma-\frac{1}{n}\}$, reveals
   \begin{align*}
     V^{liq}_{\varrho_n}(\hvp) & =\hvp^0_{\varrho_n} + (1-\lambda)\hvp^1_{\varrho_n}S_{\varrho_n} \\
      &\leq \hvp^0_\sigma -\int_\sigma^{\varrho_n}S_ud\hvp_u^{1,\uparrow} + (1-\lambda)\left(\hvp_\sigma^1 + \int_\sigma^{\varrho_n}d\hvp_u^{1,\uparrow}\right)S_{\varrho_n} \\
      &= V^{liq}_{\sigma}(\hvp) + \hvp_{\sigma}^1(1-\lambda)\underbrace{(S_{\varrho_n}-S_\sigma)}_{=-\frac{1}{n}} - \int_\sigma^{\varrho_n} \underbrace{\big(S_u-(1-\lambda)S_{\varrho_n}\big)}_{\geq S_u-S_{\varrho_n}\geq 0} d\hvp^{1,\uparrow}_u <0.
   \end{align*}
This contradiction finishes the proof of the theorem.
 \end{proof}
  To apply Theorem \ref{t1} to the fractional Black-Scholes model \eqref{fBS}, it remains to show that condition \eqref{u(x)}, requiring that the indirect utility is finite, is satisfied. This is established in the following lemma by using the estimate on the fluctuations of fractional Brownian motion from Proposition \ref{ppn:fluctuations}.
  
  Fix $H \in \, (0,1),$ $\mu\in\R$, $\sigma>0$, the fractional Black-Scholes model \eqref{fBS}, that is,
$$S_t=\exp\big(\mu t+\sigma B^H_t\big),\qquad 0\leq t\leq T.$$
Let $X_t:=\log(S_t)=\mu t+\sigma B^H_t$. For $\delta > 0$, define the \emph{$\delta$-fluctuation times $(\sigma_j)_{j \geq 0}$ of $X$} inductively by $\sigma_0 \equiv 0$ and
\[
\sigma_{j + 1} (\omega) \coloneqq \inf \big\{t \geq \sigma_j~\big|~ \abs{ X_t-X_{\sigma_j}} \geq \delta\big\}
.\]
The \emph{number of $\delta$-fluctuations of $X$ up to time  $T$} is then given by the random variable
\[
\overline{F}^{(\delta)}_T (\omega) \coloneqq \sup \{j \geq 0~|~\sigma_j (\omega) \leq T\}
.\]
  \begin{lemma}\label{fu}
Fix $H \in \, (0,1),$ $\mu\in\R$, $\sigma>0$, the fractional Black-Scholes model \eqref{fBS}, that is,
$$S_t=\exp\big(\mu t+\sigma B^H_t\big),\qquad 0\leq t\leq T,$$
as well as $\lambda >0,$ and $\delta>0$ such that $(1-\lambda) e^{2\delta} < 1.$

Then, there exists a constant $K>0,$ depending only on $\delta$ and $\lambda,$ such that, for each $\varphi^0_T \in \cC(x)$, we have
\begin{align}\label{B3}
\varphi^0_T \leq x K^n \quad \mbox{on} \quad \left\{\overline{F}^{(\delta)}_T = n\right\}.
\end{align}

In particular, $\{\E[U(\varphi^0_T)]:\varphi^0_T \in \cC(x)\}$ remains bounded from above, for any concave function $U:(0,\infty)\mapsto \mathbb{R} \cup \{-\infty\}.$
\end{lemma}

\begin{proof}
We first observe that the mapping $t\mapsto \mu t$ can at most have $\lfloor\frac{2\mu T}{\delta}\rfloor+1$ fluctuations of size $\frac{\delta}{2}$ up to time $T$, where $\lfloor x\rfloor$ denotes the largest integer less than or equal to $x\in\R$. Since $\delta=|X_{\sigma_{j+1}}-X_{\sigma_{j}}|\leq |\mu(\sigma_{j+1}-\sigma_{j})|+ \sigma|B^H_{\sigma_{j+1}}-B^H_{\sigma_{j}}|$, we therefore have that
$$\overline{F}^{(\delta)}_T\leq\textstyle \lfloor\frac{2\mu T}{\delta}\rfloor+1+F^{(\frac{\delta}{2\sigma})}_T,$$
where $F^{(\frac{\delta}{2\sigma})}_T$ denotes the number of $\frac{\delta}{2\sigma}$-fluctuations of $B^H$ up to time $T$ as defined in \eqref{dfluB}. Combining the previous estimate with Corollary \ref{cor:fluc} gives that the random variable $\overline{F}^{(\delta)}_T$ has exponential moments of all orders, that is,
\[
\E \brack[\big]{\exp (a \overline{F}^{(\delta)}_T)} < \infty\quad\text{for all $a\in\R$}\label{fluc:eq2}
.\]
As regards the final sentence of the lemma, it follows from \eqref{B3} and \eqref{fluc:eq2} that
$$\E[\vp^0_T]\leq \E\left[x K^{\overline{F}^{(\delta)}_T}\right]=x\E\left[\exp\left(\log(K)\overline{F}^{(\delta)}_T\right)\right]<\infty$$
and hence 
$\{\E[\varphi^0_T]:\varphi^0_T \in \cC(x)\}$ remains bounded. This implies the final assertion as any concave function $U$ is dominated by an affine function.

It remains to show \eqref{B3}. Fix an admissible trading strategy $\varphi$ starting at $\varphi_{0_-}=(1,0)$ and ending at $\varphi_T=(\varphi^0_T,0).$ Define the ``optimistic value'' process $(V^{opt}(\varphi_t ) )_{0 \leq t \leq T}$ by
\begin{align*}
V^{opt}(\varphi_t)=\varphi^0_t + (\varphi^1_t)^+ S_t - (\varphi^1_t)^- (1-\lambda)S_t.
\end{align*}

The difference to the liquidation value $V^{liq}$ as defined in \eqref{liq} is that we interchanged the roles of $S$ and $(1-\lambda)S.$ Clearly $V^{opt} \geq V^{liq}$. 

Fix a trajectory $(X_t(\omega))_{0 \leq t \leq T}$ of $X$ as well as $j \in \mathbb{N}$ such that $\sigma_j(\omega) < T.$ We claim that there is a constant $K=K(\lambda,\delta)$ such that, for every $\sigma_j(\omega) \leq t \leq \sigma_{j+1}(\omega) \wedge T,$

\begin{align}\label{B7.a}
V^{opt}(\varphi_t(\omega) ) \leq K V^{opt} ( \varphi_{\sigma_j}(\omega ) ).
\end{align}

To prove this claim we have to do some rough estimates. Fix $t$ as above. Note that $S_t$ is in the interval $[e^{-\delta} S_{\sigma_j}(\omega), e^\delta S_{\sigma_j} (\omega)]$ as $\sigma_j(\omega) \leq t \leq \sigma_{j+1} (\omega) \wedge T.$ To fix ideas suppose that $S_t(\omega)= e^\delta S_{\sigma_j} (\omega).$ We try to determine the trajectory $(\varphi_u)_{\sigma_j(\omega) \leq u \leq t}$ which maximises the value on the left hand side for given $V:=V^{opt}(\varphi_{\sigma_j}(\omega))$ on the right hand side. As we are only interested in an upper bound we may suppose that the agent is clairvoyant and knows the entire trajectory $(S_u(\omega))_{0 \leq u \leq T}.$ 

In the present case where $S_t(\omega)$ is assumed to be at the upper end of the interval $[e^{-\delta} S_{\sigma_j} (\omega), e^\delta S_{\sigma_j} (\omega)]$ the agent who is trying to maximize $V^{opt} (\varphi_t(\omega))$ wants to exploit this up-movement by investing into the stock $S$ as much as possible. But she cannot make $\varphi^1_u \in \mathbb{R}_+$ arbitrarily large as she is restricted by the admissibility condition $V^{liq}_u \geq 0$ which implies that $\varphi^0_u + \varphi^1_u(1-\lambda) S_u(\omega) \geq 0,$ for all $\sigma_j(\omega) \leq u \leq t.$ As for these $u$ we have $S_u(\omega) \leq e^\delta S_{\sigma_j}(\omega)$ this implies the inequality
\begin{align}\label{B8.}
\varphi^0_u + \varphi^1_u (1-\lambda) e^\delta S_{\sigma_j} (\omega) \geq 0, \quad \sigma_j (\omega) \leq u \leq t.
\end{align}
 
As regards the starting condition $V^{opt}(\varphi_{\sigma_j}(\omega))$ we may assume w.l.o.g.~that $\varphi_{\sigma_j} (\omega)=(V,0)$ for some number $V>0.$ Indeed, any other value of $\varphi_{\sigma_j} (\omega)=(\varphi^0_{\sigma_j} (\omega), \varphi^1_{\sigma_j} (\omega))$ with $V^{opt}(\varphi_{\sigma_j} (\omega))=V$ may be reached from $(V,0)$ by either buying stock at time $\sigma_j (\omega)$ at price $S_{\sigma_j} (\omega)$ or selling it at price $(1-\lambda) S_{\sigma_j} (\omega).$ Hence we face the elementary deterministic optimization problem of finding the trajectory $(\varphi^0_u, \varphi^1_u)_{\sigma_j (\omega) \leq u \leq t}$ starting at $\varphi_{\sigma_j} (\omega)=(V,0)$ and respecting the self-financing condition \eqref{eq:sf} as well as inequality \eqref{B8.}, such that it maximizes $V^{opt}(\varphi_t).$
Keeping in mind that $(1-\lambda) < e^{-2\delta},$ a moment's reflection reveals that the best (clairvoyant) strategy is to wait until the moment $\sigma_j(\omega) \leq \bar{t} \leq t$ when $S_{\bar{t}} (\omega)$ is minimal in the interval $[\sigma_j(\omega), t],$ then to buy at time $\bar{t}$ as much stock as is allowed by the inequality \eqref{B8.}, and then keeping the positions in bond and stock constant until time $t.$ Assuming the most favourable (limiting) case $S_{\bar{t}} (\omega) = e^{-\delta} S_{\sigma_j}(\omega),$ simple algebra gives $\varphi_u=(V,0),$ for $\sigma_j(\omega) \leq u < \bar{t}$ and 
$$\varphi_u=\Big(V-v, \frac{v e^\delta}{S_{\sigma_j}(\omega)}\Big), \quad \bar{t} \leq u \leq t,$$
where
$$v= \frac{V}{1-(1-\lambda) e^{2\delta}}.$$

Using $S_t(\omega) = e^\delta S_{\sigma_j} (\omega)$ we therefore may estimate in \eqref{B7.a}
\begin{align}\label{B10.a}
V^{opt}(\varphi_t(\omega) ) \leq V\Big[\Big(1-\frac{1}{1-(1-\lambda ) e^{2\delta}}\Big) + \frac{e^{2\delta}}{1-(1-\lambda ) e^{2\delta}}\Big].
\end{align}
Due to the hypothesis $(1-\lambda) e^{2\delta} < 1$ the term in the bracket is a finite constant $K$, depending only on $\lambda$ and $\delta.$
We have assumed a maximal up-movement $S_t(\omega)= e^\delta S_{\sigma_j} (\omega).$ The case of a maximal down-movement $S_t(\omega)=e^{-\delta} S_{\sigma_j} (\omega)$ as well as any intermediate case follow by the same token yielding again the estimate \eqref{B7.a} with the same constant given by \eqref{B10.a}. Observing that $V^{opt} \geq V^{liq}$ and $V^{liq}_T(\hvp)=\hvp^0_T$ we obtain inductively \eqref{B3} thus finishing the proof.
\end{proof}

\begin{proof}[Proof of Theorem \ref{t2}]
 The fractional Black-Scholes model satisfies $(TWC)$ by the law of iterated logarithm for fractional Brownian motion at stopping times in Theorem 1.1 of~\cite{Pey15}. The finiteness of the indirect utility function $u(x)<\infty$ follows from Lemma \ref{fu}. Therefore, the assumptions of Theorem \ref{t1} are satisfied and we obtain the existence of an optimal trading strategy $\hvp=(\hvp^0_t,\hvp^1_t)_{0_-\leq t\leq T}$ for \eqref{P1} and a shadow price $\widehat{S}=(\widehat{S}_t)_{0\leq t\leq T}$ in the sense of Definition \ref{def:sp}. 
\end{proof}

\begin{proof}[Proof of Theorem \ref{t3}]
We obtain the existence of an optimal trading strategy and a shadow price from Theorem \ref{t2}. The assertion that $\widehat{S}=(\widehat{S}_t)_{0\leq t\leq T}$ follows as in Lemma 5.1 of \cite{CS15} by combining that the dual optimiser is $\widehat{Y}=(\widehat{Y}^0_t,\widehat{Y}^1_t)_{0\leq t\leq T}$ is a local martingale by the proof of Theorem \ref{t1} with the predictable representation property of $W=(W_t)_{0\leq t\leq T}$ conditional on $\F_0$. Because of part 3) of the Definition \ref{def:sp} of a shadow price, the fact that the optimal fraction of wealth for the frictionless logarithmic utility maximisation  problem for \eqref{t3:ito} is given by 
$$\widehat{\pi}_t=\frac{\widehat{\mu}_t}{\widehat{\sigma}^2_t}, \quad 0\leq t\leq T,$$
implies the relation \eqref{t3:opt} between the coefficients $\widehat{\mu}=(\widehat{\mu}_t)_{0 \leq t \leq T}$ and $\widehat{\sigma} = (\widehat{\sigma}_t)_{0 \leq t \leq T}$ of the It\^o process \eqref{t3:ito} and the optimal trading strategy $\hvp=(\hvp^0_t,\hvp^1_t)_{0_-\leq t\leq T}$ for \eqref{P1:log}.
\end{proof}

\bibliography{SPfBm-2016-08-03}
\bibliographystyle{abbrv}

\end{document}